%% file: main.tex
\documentclass[11pt]{article}
\usepackage{tikz}
\usepackage{amsmath, amssymb, amsthm, tikz, thm-restate}
\usepackage{color}
\usepackage{algorithm}
\usepackage[noend]{algpseudocode}

\usepackage{verbatim}
\usetikzlibrary{arrows.meta,shapes}

\pgfdeclarelayer{background}
\pgfsetlayers{background,main}

\def\showauthornotes{1}

\def\showdraftbox{0}

\input{macros}

\allowdisplaybreaks[1]

\ifnum\showauthornotes=1
\newcommand{\todo}[1]{\colorbox{Mygray}{\color{red}\parbox{\textwidth}{#1}}}
\else
\newcommand{\todo}[1]{}
\fi

\newcommand{\eps}{\varepsilon}

\newcommand{\lap}{\ensuremath{\boldsymbol{L}}}

\newcommand\Otil[1]{\ensuremath{\widetilde{O}\left(#1\right)}}

\newcommand\LL{\boldsymbol{\mathit{L}}}

\newcommand\xx{\boldsymbol{\mathit{x}}}
\newcommand\vv{\boldsymbol{\mathit{v}}}

\newcommand{\DegPreserveSparsify}{{\textsc{DegreePreservingSparsify}}}
\newcommand{\SpectralSketch}{{\textsc{SpectralSketch}}}

\newcommand{\NaiveCycleDecomposition}{{\textsc{NaiveCycleDecomposition}}}

\newcommand{\CycleDecomposition}{{\textsc{CycleDecomposition}}}

\newcommand{\EulerianSparsify}{{\textsc{EulerianSparsify}}}

\newcommand{\dir}[1]{{\vec{#1}}}

\newcommand{\etal}{\emph{et al.}}

\newcommand{\extra}{\mathrm{\textbf{extra}}}

\newcommand{\arrow}{\vec}

\newcommand{\LowDiamDecomp}{{\textsc{LowDiamDecomp}}}
\newcommand{\ShortCycleDecomp}{{\textsc{ShortCycleDecomp}}}
\newcommand{\TreeSplit}{{\textsc{TreeSplit}}}
\newcommand{\NaiveShortCycle}{{\textsc{NaiveShortCycle}}}
\newcommand{\Sparsify}{{\textsc{Sparsify}}}
\newcommand{\PullUp}{{\textsc{PullUp}}}
\renewcommand{\hat}{\widehat}
\renewcommand{\tilde}{\widetilde}
\newcommand{\ImprovedShortCycle}{{\textsc{ImprovedShortCycle}}}
\newcommand{\diam}{\mathrm{diam}}
\newcommand{\OneRoundShortCycle}{{\textsc{OneRoundShortCycle}}}

\newcommand{\GraphReduce}{\textsc{GraphReduce}}
\newcommand{\SparsifyHelper}{\textsc{SparsifyHelper}}

\title{%
Short Cycles via Low-Diameter Decompositions
}%

\author{
  Yang P. Liu\thanks{Stanford University.
    \texttt{yangpatil@gmail.com}. Research supported by the U.S.
Department of Defense via an NDSEG fellowship.}
  \and
  Sushant Sachdeva\thanks{ University of Toronto.
    \texttt{sachdeva@cs.toronto.edu}. Research supported in part by the
Natural Sciences and Engineering Research Council of Canada (NSERC),
and a Connaught New Researcher award.}
  \and
  Zejun Yu\thanks{University of Waterloo.
    \texttt{z248yu@uwaterloo.ca}. This work was done when this author
    was an undergrad student at the University of Toronto.}
  }

\date{\today}
\begin{document}

\maketitle
\thispagestyle{empty}

\input{abstract}

\newpage
\setcounter{page}{1}

\input{intro}
\input{prelim}

\input{reduc}

\input{nrtn}

\input{proofs}

{\small
\bibliographystyle{alpha}
\bibliography{refs}}
\pagebreak
\begin{appendix}
\input{appendix}
\end{appendix}

\end{document}

%% file: macros.tex
\usepackage{geometry}
\usepackage[pagebackref]{hyperref}
\usepackage{amsmath,amssymb,amsthm,amstext,amsfonts,bbm,algorithm,algorithmicx,graphicx,xspace,nicefrac}
\usepackage{color,stmaryrd,enumerate,latexsym,bm,amsfonts,wrapfig,verbatim,tabularx,textcomp,
subfig}
\usepackage[square,sort,comma,numbers]{natbib}

\usepackage{amsfonts}
\usepackage{comment} 
\usepackage{epsfig} 
\usepackage{latexsym,nicefrac,bbm}
\usepackage{xspace}
\usepackage{color,fancybox,graphicx,url}
\usepackage{enumitem, fullpage}
\usepackage{booktabs}
\usepackage{commath}
\usepackage{mdframed}
\usepackage{pdfsync}

\usepackage{fancybox}
\newenvironment{fminipage}%
  {\begin{Sbox}\begin{minipage}}%
  {\end{minipage}\end{Sbox}\fbox{\TheSbox}}

\newcommand{\defeq}{\stackrel{\textup{def}}{=}}

\newtheorem{theorem}{Theorem}[section]

\newtheorem{lemma}[theorem]{Lemma}
\newtheorem{definition}[theorem]{Definition}

\newcommand{\nfrac}[2]{\nicefrac{#1}{#2}}

\newcommand{\ceil}[1]{\left\lceil\, {#1}\,\right\rceil}

\newcommand\rea{\mathbb R}

\newcommand{\marginlabel}[1]%
{\mbox{}\marginpar{\it{\raggedleft\hspace{0pt}#1}}}

\definecolor{Mygray}{gray}{0.8}

 \ifcsname ifcommentflag\endcsname\else
  \expandafter\let\csname ifcommentflag\expandafter\endcsname
                  \csname iffalse\endcsname
\fi

\ifnum\showauthornotes=1
\newcommand{\Authornote}[2]{{\sf\small\color{red}{[#1: #2]}}}
\newcommand{\Authoredit}[2]{{\sf\small\color{red}{[#1]}\color{blue}{#2}}}
\newcommand{\Authorcomment}[2]{{\sf \small\color{gray}{[#1: #2]}}}
\newcommand{\Authorfnote}[2]{\footnote{\color{red}{#1: #2}}}
\newcommand{\Authorfixme}[1]{\Authornote{#1}{\textbf{??}}}
\newcommand{\Authormarginmark}[1]{\marginpar{\textcolor{red}{\fbox{
#1:!}}}}
\else
\newcommand{\Authornote}[2]{}
\newcommand{\Authoredit}[2]{}
\newcommand{\Authorcomment}[2]{}
\newcommand{\Authorfnote}[2]{}
\newcommand{\Authorfixme}[1]{}
\newcommand{\Authormarginmark}[1]{}
\fi

\newcommand{\paren}[1]{\left({#1}\right)}

\newlength{\pgmtab}  
\newcommand {\assign}{\leftarrow}

 {
  \begin{enumerate}}{\end{enumerate}}

\def\qedsketch{\ifmmode\Box\else{\unskip\nobreak\hfil
\penalty50\hskip1em\null\nobreak\hfil$\Box$
\parfillskip=0pt\finalhyphendemerits=0\endgraf}\fi}

\newlength{\tpush}
\newcommand{\handout}[5]{
   \noindent
   \begin{center}
   \framebox{ \vbox{ \hbox to \textwidth { {\bf \coursenum\ :\  \coursename} \hfill #5 }
       \vspace{3mm}
       \hbox to \textwidth { {\Large \hfill #2  \hfill} }
       \vspace{1mm}
       \hbox to \textwidth { {\it #3 \hfill #4} }
     }
   }
   \end{center}
   \vspace*{4mm}
   \newcommand{\lecturenum}{#1}
   \addcontentsline{toc}{chapter}{Lecture #1 -- #2}
}

\ifnum\showdraftbox=1

\else

\fi



%% file: abstract.tex
\begin{abstract}{
  We present improved algorithms for \emph{short cycle decomposition}
  of a graph -- a decomposition of an undirected, unweighted graph
  into edge-disjoint cycles, plus a small number of additional
  edges. Short cycle decompositions were introduced in the recent work
  of Chu~{\etal} (FOCS 2018), and were used to make progress on several
  questions in graph sparsification.}

  {
  For all constants $\delta \in (0,1]$, we
  give an $O(mn^\delta)$ time algorithm that, given a graph $G,$
  partitions its edges into cycles of length
  $O(\log n)^\frac{1}{\delta}$, with $O(n)$ extra edges not in any
  cycle. This gives the first subquadratic, in fact almost linear
  time, algorithm achieving polylogarithmic cycle lengths.  We also
  give an $m \cdot \exp(O(\sqrt{\log n}))$ time algorithm that
  partitions the edges of a graph into cycles of length
  $\exp(O(\sqrt{\log n} \log\log n))$, with $O(n)$ extra edges not in
  any cycle.  This improves on the short cycle decomposition
  algorithms given by Chu~\emph{et al.} in terms of all parameters, and is
  significantly simpler.}

  {  As a result, we obtain faster algorithms
    and improved guarantees for several problems in graph
    sparsification -- construction of resistance sparsifiers,
    graphical spectral sketches, degree preserving sparsifiers, and
    approximating the effective resistances of all edges.}
\end{abstract}


%% file: intro.tex
\section{Introduction}

Graph sparsification is the problem of approximating a graph $G$ by a
sparse graph $H,$ while preserving some key properties of the
graph. Several notions of graph sparsification have been studied. For
instance, graph spanners introduced by Chew~\cite{Chew89}
approximately preserve distances, and cut-sparsifiers introduced by
Benczur and Karger~\cite{BenczurK96} approximately preserve the sizes
of all cuts.

The notion of spectral sparsification defined by Spielman and
Teng~\cite{ST11b, SpielmanT04} approximately preserves the Laplacian
quadratic form of the graph. To define a spectral sparsifier, we
recall the definition of the Laplacian of a graph.  For an undirected,
weighted graph $G = (V, E_G, w_G),$ with $n$ vertices and $m$ edges,
the Laplacian of $G,$ $\lap_G$ is the unique symmetric $n \times n$
matrix such that for all $\xx \in \rea^{n},$ we have
\[
  \xx^{\top}\lap_G \xx = \sum_{(u,v) \in E_G} w_G(u,v)(\xx_u - \xx_v)^{2}.
\]
For $\eps < 1,$ a graph $H$ is said to be an $\eps$-sparsifier for $G$
if we have
\[\forall  \xx \in \rea^{n}, \quad(1-\eps)\xx^{\top}\LL_{G}\xx \le \xx^{\top}\LL_{H}\xx \le
  (1+\eps)\xx^{\top}\LL_{G}\xx.
\]
Considering $\xx$ as indicator vectors of a cut shows that a spectral
sparsifier is also a cut sparsifier.

Spectral sparsifiers have found numerous applications in graph
algorithms -- they are a crucial component of several fast solvers
for Laplacian linear systems (this was the main objective of Spielman
and Teng)~\cite{SpielmanT04, SpielmanTengSolver:journal, KoutisMP10,
  KoutisMP11}. Additionally, they are the \emph{only} graph theoretic primitive in some
of them~\cite{PengS14, KyngLPSS16}, such as faster cut and flow
algorithms~\cite{Sherman13, Sherman09, ChristianoKMST10}, sampling
random spanning trees~\cite{DurfeeKPRS17}, estimating
determinants~\cite{DurfeePPR17}, etc.

Spectral sparsification is widely considered to be well understood:
Following a sequence of works~\cite{ST11b, SpielmanS08:journal},
Batson, Spielman, and Srivastava~\cite{BSS12, BatsonSS09} showed how
to construct graph sparsifiers with $O(n\eps^{-2})$ edges. We also
know that this bound is tight for graphs~\cite{BSS12}, and even for
arbitrary data-structures that can answer the sizes of all cuts up to
$(1\pm\eps)$ factors~\cite{CarlsonKNT17:arxiv}.

However, a sequence of recent works~\cite{AndoniCKQWZ16, DinitzKW15,
  JambulapatiS18,CohenKPPRSV17} opened up several interesting new
directions and open questions in spectral sparsification:
\begin{enumerate}
\item Building on the work of Andoni~\etal~\cite{AndoniCKQWZ16},
  Jambulapati and Sidford~\cite{JambulapatiS18}, showed how to
  construct data structures (\emph{spectral sketches}) with
  $\Otil{n\eps^{-1}}$ space that can estimate the quadratic form
  $\xx^{\top}\LL_{G}\xx$ for a \emph{fixed} unknown vector
  $\xx \in \rea^{n}$ with high probability,

  even though any data-structure that can answer queries even for all
  $\xx \in \{\pm 1\}^{n}$ needs $\Omega(n\eps^{-2})$
  space~\cite{AndoniCKQWZ16}. Do there exist graphs with
  $\Otil{n\eps^{-1}}$ edges that are spectral sketches?
\item Dinitz~\etal~\cite{DinitzKW15} showed that for expander graphs,
  there exist \emph{resistance sparsifiers} with $\Otil{n\eps^{-1}}$
  edges. Resistance sparsifiers preserve the effective
  resistance\footnote{The effective resistance between $u,v$ is the
    potential difference between $u,v$ if the graph is considered an
    electrical network with edge $e$ with weight $w_e$ as a resistor
    with resistance $1/w_e,$ and a unit current is sent from $u$ to
    $v.$} between all pairs of vertices up to $(1\pm\eps).$
  Dinitz~\etal~conjecture that all graphs have resistance sparsifiers
  with $\Otil{n\eps^{-1}}$ edges.
\end{enumerate}

A recent work of Chu~\etal~\cite{ChuGPSSW18} answered both the above
questions affirmatively, giving the first constructions of graphical
spectral sketches and resistance sparsifiers for all graphs with
$\Otil{n\eps^{-1}}$ edges.

A key component of their algorithms is a novel decomposition of graphs
-- a short-cycle decomposition --
into short edge-disjoint cycles and a few extra edges.
\begin{definition}{Short Cycle Decomposition. \cite{ChuGPSSW18}}
  A $(\hat{k}, L)$-short cycle decomposition of an unweighted
  undirected graph $G$ decomposes $G$ into several edge-disjoint
  cycles, each of length at most $L$, and at most $\hat{k}$ edges not
  in these cycles.
\end{definition}

In addition to resolving the above two open problems, Chu~\etal~also
show that short cycle decompositions can be used for constructing
spectral sparsifiers that preserve degrees, sparsifying Eulerian
directed graphs (directed graphs with all vertices having in-degree
equal to out-degree), and faster estimation of effective resistances.

The existence of $(2n,2\log n)$-short cycle decompositions follows
from a simple observation that every graph with minimum degree 3 must
have a cycle of length $2 \log n,$ which can be found by a simple
breadth-first search. If a graph has more than $2n$ edges, iteratively
removing vertices of degree at most 2 must leave a graph with
min-degree 3, and hence the graph contains a short cycle. Removing
this cycle and repeating gives a simple $O(mn)$ time algorithm for
producing a $(2n, 2\log n)$-short cycle decomposition of a graph $G$.
It is described as \NaiveCycleDecomposition~in~\cite[Algorithm
11]{ChuGPSSW18}.

In order to give almost-linear time algorithms for their applications,
Chu~\etal~\cite{ChuGPSSW18} give an algorithm
{\textsc{ShortCycleDecomposition}} \cite[Algorithm 15]{ChuGPSSW18}
that runs in time $m \cdot \exp(O(\log n)^{\nfrac{3}{4}}),$ and
returns a
$(n \exp(O(\log n \log\log n)^{\nfrac{3}{4}}), \exp(O(\log
n)^{\nfrac{3}{4}}))$-short cycle decomposition of the graph
(see~\cite[Theorem 3.11]{ChuGPSSW18}).

\subsection{Our Contributions.}
Our main result is a new algorithm for short cycle decomposition,
which improves over the algorithms in the work of Chu
\etal~\cite{ChuGPSSW18} in terms of all parameters, is faster, and
considerably simpler.
\begin{theorem}
  \label{thm:runtimeandreduction}
  For all integers $c \ge 1$, we give an algorithm that, given a graph
  $G$ with $n$ vertices and $m$ edges, runs in time
  $O\left(mn^\frac{1}{c+1} \cdot 500^c \right),$ and returns a
  $(O(n), O(\log n)^c)$-short cycle decomposition of $G$ with high
  probability.
\end{theorem}

As immediate consequences the above theorem, we obtain improvements on
several of the results from \cite{ChuGPSSW18}. Throughout, we let $G$
be a graph with $n$ vertices and $m$ edges, and assume that the
algorithms mentioned below are run using our algorithm {\ShortCycleDecomp}
(Algorithm \ref{algo:ShortCycle}) as its
{\CycleDecomposition} algorithm.

We obtain improved degree-preserving sparsifiers by plugging in
Theorem \ref{thm:runtimeandreduction} in~\cite[Theorem
4.1]{ChuGPSSW18}.
\begin{theorem}[Degree-Preserving Sparsification]
  For any integer $c \ge 1$, algorithm {\DegPreserveSparsify} from
  \cite{ChuGPSSW18} returns a graph $H$ with at most
  $n\eps^{-2} \cdot \paren{O(\log n)}^{c+1}$ edges such that with high
  probability all vertices have the same weighted degrees in $G$ and
  $H$, $H$ is an $\eps$-spectral sparsifier of $G.$ The algorithm runs
  in time $\tilde{O}(500^c \cdot m \cdot n^\frac{1}{c+1}).$
\end{theorem}

Combining Theorem \ref{thm:runtimeandreduction} with~\cite[Theorem
6.1]{ChuGPSSW18} gives an improved construction of graphical spectral
sketches and resistance sparsifiers.
\begin{theorem}
  For any integer $c \ge 1$, algorithm {\SpectralSketch} from
  \cite{ChuGPSSW18}, given an undirected weighted graph $G$ and
  parameter $\eps$ as inputs, runs in time
  $\tilde{O}(500^c \cdot m \cdot n^\frac{1}{c+1}),$ and returns a
  graph $H$ with
  $\widetilde{O}(n\eps^{-1}) \cdot \paren{O(\log n)}^{c+1}$ edges such
  that with high probability
  \begin{enumerate}
  \item $H$ is an $\eps$-spectral sketch for $G,$ i.e., for any fixed
    vector $\xx$, with high probability
    $\xx^{\top} \LL_H \xx = (1\pm\eps) \xx^{\top} \LL_G \xx.$
  \item $H$ is an $\eps$-resistance sparsifier for $G.$ In fact, for
    any fixed vector $\xx$, with high probability,
    $\xx^{\top} \LL_H^{+} \xx = (1\pm\eps) \xx^{\top} \LL_G^{+} \xx.$
    \footnote{$\LL^{+}$ denotes the Moore-Penrose pseudoinverse of
      $\LL.$ If the eigendecomposition of $\LL$ is
      $\sum_i \lambda_i \vv_i \vv_i^{\top},$ we have
      $\LL^{+} = \sum_{i:\lambda_i > 0} \frac{1}{\lambda_i} \vv_i
      \vv_i^{\top}$}
\end{enumerate}
\end{theorem}

Following the proof of Theorem 3.8 in Chu~\etal~\cite{ChuGPSSW18}
while applying our Theorem \ref{thm:runtimeandreduction} gives an
improved algorithm for estimating the effective resistances between
all pairs.
\begin{theorem}
  Given an undirected graph $G$ with $n$ vertices, $m$ edges, and any
  $t$ vertex pairs and error $\eps > 0$, we can with high probability
  compute $\eps$-approximations to the effective resistances between
  all $t$ of these pairs in time
  $\tilde{O}(m + (n+t)\eps^{-1.5}) \exp(O(\sqrt{\log n \log\log n})).$
\end{theorem}
In contrast, Theorem~3.8 from Chu~\etal~\cite{ChuGPSSW18} has a
running time of
$\tilde{O}(m + (n+t)\eps^{-1.5}) \exp(O(\log n)^{\nfrac{3}{4}}).$

Finally, plugging in Theorem \ref{thm:runtimeandreduction} in~\cite[Theorem
5.1]{ChuGPSSW18} allows us to give an algorithm for sparsifying
Eulerian directed graphs. Note that while this result is worse than
the $\Otil{m}$ algorithm given by Cohen~\etal~\cite{CohenKPPRSV17},
this gives the first almost-linear time algorithm (for
$c = \sqrt{\log n}$) for sparsifying Eulerian directed graphs that
does not require expander decompositions.
\begin{theorem}
  For any integer $c \ge 1$, algorithm {\EulerianSparsify} from
  \cite{ChuGPSSW18}, given an Eulerian directed graph $\arrow{G}$ with
  poly bounded edge weights as input, runs in time
  $\tilde{O}(500^c \cdot m \cdot n^\frac{1}{c+1}),$ and returns an
  Eulerian directed graph $\arrow{H}$ with at most
  $n\eps^{-2} \cdot \paren{O(\log n)}^{3c+1}$ edges such that with
  high probability \footnote{For a directed graph $\dir{G}$, its
    directed Laplacian, $\LL_{\dir{G}}$, can be defined as
\[
\LL_{\dir{G}}(u,v)
:=
\begin{cases}
\text{out-degree of $u$}
&
\qquad
\text{if $u = v$,}
\\
-\text{(weight of edge $v \rightarrow u$)}
&
\qquad
\text{if $u \neq v$ and} \\ 
&
\qquad
\text{$v \rightarrow u$ is an edge.}
\end{cases}
\]}
  \[
    \left\|\LL_G^{+/2}\left(\LL_{\arrow{G}} -
        \LL_{\arrow{H}}\right)\LL_G^{+/2} \right\|_2 \le \eps.
  \]
\end{theorem}

\paragraph{Comparison to the work of Chu~\etal~\cite{ChuGPSSW18}.}

Setting $c = 1$ in Theorem \ref{thm:runtimeandreduction} gives an
algorithm that finds an $(O(n), O(\log n))$-short cycle in time
$O(m\sqrt{n})$ (in comparison with $O(mn)$ time for such a
decomposition in~\cite{ChuGPSSW18}).  Setting
$c = \nfrac{1}{\delta} - 1,$ for $\delta \in (0,\nfrac{1}{2}]$ gives
an us an $O(mn^\delta)$ time algorithm that finds an
$(O(n), O(\log n)^{\nfrac{1}{\delta}-1})$-short cycle
decomposition. On the other hand, the approach of
Chu~\etal~\cite{ChuGPSSW18} can only achieve sub-quadratic time if
their cycles are length at least $\exp(\sqrt{\log n\log\log n})$ (see paragraph
below for discussion).  Setting $c = \sqrt{\log n}$ in
Theorem~\ref{thm:runtimeandreduction}, we obtain an algorithm that
runs in $m \cdot \exp\left(O(\sqrt{\log
    n})\right)$
time, and finds a $(O(n), \exp(O(\sqrt{\log n}\log\log n)))$-short
cycle decomposition of the
graph.
This beats the algorithms from Chu~{\etal} in terms of all parameters:
runtime, cycle length, and extra edges. Note that these improvements
carry over immediately to all applications.

Moreover, our algorithm is considerably simpler than that of
\cite{ChuGPSSW18}. The algorithm in \cite{ChuGPSSW18} requires a
strong version of an expander decomposition instead of the more
standard expander decomposition algorithm of Spielman and
Teng~\cite{ST11b}. Instead of each piece of the decomposition being
contained in expanders, they need for each piece itself to be an
expander. This requires a stronger expander decomposition which is
given in work of Nanongkai and Saranurak \cite{NS17}. This immediately
gives an overhead of $\exp\paren{O(\sqrt{\log n \log \log n})}$ on the
lengths of cycles produced by their algorithm, even if the recursion
depth is set to some small integer constant. Our algorithm bypasses
this by using only low diameter decomposition~\cite{Bartal96}, which
allows us to generate cycles of length $\paren{O(\log n)}^{c}$ for any
constant $c \ge 1.$

\paragraph{Discussion of the work of Parter and Yogev \cite{PY17}.}
Parter and Yogev~\cite{PY17} study a closely related notion of a
\emph{low-congestion cycle cover} -- a collection of short cycles that
covers all the edges of a graph, and where each edge appears only in a
small number of cycles. An efficient construction of a low congestion
cycle cover would immediately imply an efficient algorithm for short
cycle decomposition. The methods and results presented in the paper
are very interesting. However, to be best of our knowledge, the
algorithms described in their paper only lend themselves to quadratic
time implementations.

%% file: prelim.tex
\section{Preliminaries}

Throughout we work with undirected unweighted multigraphs, allowing for multiple edges and self-loops. We say that self-loops add degree $2$ to the vertex it is attached to. \\

\noindent In this work, we often work with vertex disjoint cycles instead of edge disjoint cycles.
\begin{definition}
A \emph{vertex disjoint} short cycle decomposition is a short cycle decomposition where no two of the cycles share a vertex.
\end{definition}

\noindent For a graph $G$, let $V(G)$ and $E(G)$ denote the vertex and edge sets of $G$. For a subgraph $S \subseteq G$, define $V(S)$ to be the set of vertices of $S$, and $E(S)$ the set of edges. Generally when the graph $G$ is clear from context, we let $n$ and $m$ denote $|V(G)|$ and $|E(G)|$ respectively. \\

\noindent For a graph $G$, let $\Delta \defeq \Delta(G)$ denote the
maximum degree of the graph $G$. \\

\noindent For a subgraph $G' \subseteq G$ (possibly with $V(G') \neq V(G)$), let the (strong) diameter of $G'$ be the maximum distance between
two vertices in $V(G')$ using only the edges in $E(G')$. \\

\noindent For disjoint subsets of vertices $A_1, A_2, \dots, A_k$ of a graph, let $E(A_1, \dots, A_k)$ denote the set of edges in with endpoints in different $A_i.$ \\

\subsection{Contraction.}
\label{sec:contraction}

In this section, we discuss contraction of components in a graph, which plays a major role in our algorithms.

Let $G$ be a graph with $n$ vertices and $m$ edges, and let $A_1, A_2, \dots, A_k$ be a partition of its vertices into disjoint components.
Define the \emph{contraction} of the components $A_1, A_2, \dots, A_k$ to be the following graph, which we call $H$. $H$ has $k$ vertices numbered $1, 2, \dots, k$, where vertex $i$ corresponds to component $A_i$ in $G$. Now, for each edge $uv \in E(G)$, if $u \in A_{u'}$ and $v \in A_{v'}$, add edge $u'v'$ to $H$. There is a clear bijection between the edges of $G$ and the edges of $H$, hence $H$ has $m$ edges too.

Define the \emph{edge injection} $f: E(H) \to E(G)$ to be the function
naturally obtained from the bijection described above.

%% file: reduc.tex
\section{Reduction to Sparse, Approximately Regular Graphs}
\label{sec:reduc}

In this section, we demonstrate that it suffices to provide algorithms for graphs $G$ which are both sparse (i.e. $m = O(n)$) and have bounded degree ($\Delta(G) = O(1)$).
\begin{lemma}
\label{lemma:graphreduce} {\GraphReduce} (Algorithm \ref{algo:graphreduce})
runs in $O(m+n)$ time and takes any
graph $G$ with $n$ vertices and $m \ge n$ edges, and returns a graph
$H$, such that:
\begin{enumerate}
\item $H$ has at most $2n$ vertices and exactly $m$ edges.
\item $\Delta(H) \le \ceil{\frac{2m}{n}}$.
\item A $(\hat{k}, L)$-short cycle decomposition of $H$ can be mapped
  in $O(m+n)$ time to a $(\hat{k}, L)$-short cycle decomposition of $G$.
\end{enumerate}
\end{lemma}

\begin{proof}
  Consider the {\GraphReduce} (Algorithm \ref{algo:graphreduce}). It
  can be implemented to run in $O(m+n)$ time.

\begin{algorithm}[h]
  \caption{\GraphReduce, takes a graph $G$ with $n$ vertices,
    $m \ge n$ edges. Returns a graph $H$ with mapping of vertices from
    $H$ to $G$ where $\Delta(H) \le \ceil{\frac{2m}{n}}$.}
\begin{algorithmic}[1]
\State $D \assign \ceil{\frac{2m}{n}}$
\State Initialize $H$ to be the same as graph $G$
\For{vertex $v' \in V(G)$}
  \State Let $v$ be the corresponding vertex in $V(H)$
	\State $t \assign \left\lceil\frac{\deg(v)}{D}\right\rceil$
	\State Let the neighbors of $v$ in $H$ be $u_1, u_2, \dots, u_{\deg(v)}$
  \State Delete $v$ from $H$
	\For{$i=1$ to $t$}
		\State Add a new vertex to $H$ connected to $u_{1+(i-1)D}, u_{2+(i-1)D}, \dots, u_{\min(\deg(v), iD)}$
	\EndFor
\EndFor
\State \Return $H$ and the vertex mapping (which vertices of $H$ come from which vertices in $G$)
\end{algorithmic}
\label{algo:graphreduce}
\end{algorithm}

Clearly when the algorithm ends, all vertices in $H$ have degree at
most $\frac{2m}{n}$, as $\deg(v)-(t[v]-1)D \le D$. Also
notice that the number of vertices in $H$ is
\begin{align*}
  \sum_v t[v] &= \sum_v \left\lceil\frac{\deg(v)}{D} \right\rceil\\ &\le n
  + \sum_v \frac{\deg(v)}{D}\\ &= n + \frac{2m}{D}\\ &\le n +
  \frac{2m}{2m/n} = 2n.
\end{align*}
  So $H$ has at most $2n$ vertices as
desired.  There is a natural mapping from the vertices and edges of
$H$ to the vertices and edges of $G$ respectively. This mapping allows
us to map each cycle in $H$ to a circuit in $G$ with identical
length. Note that this circuit might now visit vertices more than
once, but we can split a circuit into cycles with the same total length
in time linear in the length of the cycle. This allows us to
efficiently map a $(\hat{k}, L)$-short cycle decomposition of $H$ to a
$(\hat{k}, L)$-short cycle decomposition of $G$. When we split the
vertices, any edge is visited exactly twice, so the algorithm takes
$O(m + n)$ time.
\end{proof}

We are now ready to state our main reduction result.

\begin{lemma}
\label{lemma:reduction}
Assume that an Algorithm $A$ takes as input graphs $G$ with $n$
vertices, $m = 10n$ edges, and maximum degree $\Delta$, and returns
vertex disjoint cycles of length $L(n)$ containing at least
$\Omega\left(\frac{n}{\Delta}\right)$ total vertices, in time $T(n).$

Then, we can construct another Algorithm $B$ that takes as input a
graph $G$ with $n$ vertices and $m$ edges, and outputs a
$(20n, L(n))$-short cycle decomposition and runs in time
$O\left(\frac{m \cdot T(n)}{n} \right).$
\end{lemma}

\begin{proof}
Algorithm $B$ operates as described in Algorithm \ref{algo:algob}.

\begin{algorithm}[h]
\caption{Algorithm $B$, takes a graph $G$ with $n$ vertices and $m$ edges and outputs a $(20n, L(n))$-short cycle decomposition}
\begin{algorithmic}[1]
  \State Let $C \assign \emptyset$ (the set of cycles we've found).
  \While {$G$ has more than $20n$ edges}
\State Let $G'$ be any subgraph of $G$ with exactly $20n$ edges.
\State Let $H \assign$ \GraphReduce$(G')$.
\State Add isolated vertices to $H$ until $H$ has exactly $2n$ vertices.
\State Let $C'$ be the set of cycles on $H$ we get when we run Algorithm $A$ on $H$.
\State Let $C''$ be the set of cycles on $G$ corresponding to $C'$.
\State Delete the edges of cycles in $C''$ from $G$.
\State Let $C \assign C \cup C''.$
\EndWhile
\State \Return $C$.
\end{algorithmic}
\label{algo:algob}
\end{algorithm}

Note that because each $G'$ has $20n$ edges and $n$ vertices, we know
that $\Delta(H) \le 40$. Additionally, graph $H$ will have exactly
$2n$ vertices and $20n$ edges. Therefore, by the conditions on
Algorithm $A$ stated, we know that the cycles of $C'$ contain at least
$\Omega\left(\frac{n}{40}\right) = \Omega(n)$ edges of $H$. Therefore,
the cycles of $C''$ also contain at least $\Omega(n)$ edges. When the
algorithm terminates, $G$ has less than $20n$ edges remaining, so it
is clear that Algorithm $B$ returns an $(20n, L(n))$-short cycle
decomposition. As each iteration of Algorithm $B$ removes cycles
containing at least $\Omega(n)$ edges, it repeats at most
$O\left(\frac{m}{n}\right)$ times, for a total runtime of
$O\left(\frac{m \cdot T(n)}{n}\right)$ as Algorithm $A$ takes $T(n)$
time and all other processing takes $O(n)$ time per iteration.
\end{proof}

The above reduction allows us to work with bounded degree graphs.
Our main algorithms {\ImprovedShortCycle} (Algorithm
\ref{algo:ImprovedShortCycle}) and {\ShortCycleDecomp} (Algorithm
\ref{algo:ShortCycle}) both satisfy the conditions of Lemma
\ref{lemma:reduction}, and assume that the input graph satisfies
$m = 10n$.


%% file: nrtn.tex
\section{Improved Naive Cycle Decomposition}
\label{sec:nrtn}

In this section we present an improvement on {\NaiveShortCycle} by
giving an algorithm (Algorithm \ref{algo:ImprovedShortCycle})
that when given a graph $G$ with $n$ vertices and
$m = 10n$ vertices, returns vertex disjoint cycles of length $O(\log n)$
containing at least $\frac{m}{10\Delta}$ vertices in total.
It runs in time $O(m\sqrt{n}).$

Before stating the algorithm, we state several of the subalgorithms which we have described in Section \ref{sec:overview}. We defer the proofs to the appendix.

The first is a vertex disjoint version of {\NaiveCycleDecomposition} from \cite{ChuGPSSW18}.
\begin{restatable}{lemma}{restatenaivecycle}
\label{naivecycle}
{\NaiveShortCycle} (Algorithm \ref{algo:NaiveShortCycle}) takes a graph $G$ with $n$
vertices, $m$ edges, and maximum degree $\Delta$, and outputs vertex
disjoint cycles of length at most $2 \log n$ containing at least
$\frac{m-2n}{\Delta}$ total vertices. It runs in $O(m + n^{2})$ time.
\end{restatable}

We will need low-diameter decompositions. They were introduced by
Bartal~\cite{Bartal96}. We use the following version from the work of
Miller~\etal~\cite{MillerPX13}.
\begin{restatable}[Theorem 1.2 from \cite{MillerPX13}]{theorem}{restatelowdiamdecomp}
\label{lowdiamdecomp}
There is an algorithm {\LowDiamDecomp}$(G, \beta)$ that takes a graph with $G$ and
a parameter $\beta$ and with high probability returns a set $R$ of
edges of $G$ of size $\beta m$ such that all connected components of
$G \backslash R$ have diameter $O(\beta^{-1} \log n).$ The algorithm
{\LowDiamDecomp} runs in time $O(m).$
\end{restatable}

Additionally, we need a simple routine to essentially ``pull back" a short cycle decomposition
on a contracted graph back to the base graph.
\begin{restatable}{lemma}{restatepullup}
\label{pullup}The algorithm
\PullUp$(H', C'$, \\$\{K_i\}_{i=1}^{|V(H')|}$, $\{S_i\}_{i=1}^{|V(H')|}$, $f)$ (Algorithm \ref{algo:PullUp}) takes the following inputs:
\begin{enumerate}
\item $H'$-- a graph.
\item $C'$-- a set of vertex disjoint cycles on the vertices of $H'$.
\item $\{K_i\}_{i=1}^{|V(H')|}$-- A partition of the vertices of a graph $G$, where $K_i$ corresponds to vertex $i$ of graph $H'$.
\item $\{S_i\}_{i=1}^{|V(H')|}$-- Each $S_i$ is a spanning tree on the vertices in $K_i$.
\item $f$-- This is a function $f: E(H') \to E(G)$ such that for an edge $uv \in E(H')$, we have that $f(uv) \in E(K_u, K_v)$.
\end{enumerate}
It returns a set $C$  of cycles on the vertices of $G$ such that
\begin{enumerate}
\item The cycles in $C$ are vertex disjoint.
\item The cycles in $C$ cover at least as many vertices as those in $C'$ did.
\item The length in $C$ have maximum length at most $O(\max_i \diam(S_i))$ times the longest cycle in $C'$.
\end{enumerate}
It runs in time $O(n).$
\end{restatable}

Finally, we need an algorithm that splits a tree into smaller subtrees.
We will use this in order to split connected components in a graph
into smaller connected components that are all approximately equal sized.
\begin{restatable}{lemma}{restatetreesplit}
\label{treesplit}
Let $T$ be a tree with $n$ vertices maximum degree $D$. Assume the vertices are
labelled with nonnegative integers between $0$ and $X$ (denote labels as
$c_v$ for $v \in V(T)$). Then the algorithm {\TreeSplit}$(n, t, T, X, \{c_v\}_{v\in V(T)})$ (Algorithm
\ref{algo:treesplit}) is an $O(n)$ time algorithm that when given $T$
and a positive integer $t \le \sum_v c_v$, splits the tree into
connected subgraphs, each of which has sum of labels between $t$ and
$Dt + X$.
\end{restatable}

Now we proceed to the algorithm, which we split into two parts. The main part is \OneRoundShortCycle, which finds many vertex disjoint cycles on a graph with diameter $O(\log n)$.

\begin{algorithm}[h]
\caption{\OneRoundShortCycle, takes a graph $G$ with $n$ vertices, $m$ edges, maximum degree $\Delta$, and diameter $O(\log n)$. Returns vertex disjoint cycles of length $O(\log n)$ containing at least $\frac{m-5n}{10\Delta\sqrt{m}}$ vertices}
\begin{algorithmic}[1]
\State Let $T$ be a spanning tree of diameter $O(\log n)$ of $G$.
\For{$v \in V(G)$} let $c_v \assign \deg v.$
\EndFor
\State Let $K \assign $\TreeSplit$(n, 4\sqrt{m}, T, \Delta, \{c_v\}_{v\in V(G)})$. \label{line:treesplitnrtn}
\State Let $K = \{K_1, K_2, \dots, K_{|K|}\}.$
\State Initialize graph $H$ on $|K|$ vertices as
empty. Each vertex $i$ will correspond to $K_i.$
\For{$1 \le i \le |K|$}
	\State Let $S_i$ be a spanning tree of $K_i$ of diameter $O(\log n)$.
\EndFor
\State Let $H$ be the graph obtained by contracting the components
$K_1, K_2, \dots, K_{|K|}$
\State Remove the edges in $H$ corresponding to the edges in the trees
$S_1, \dots, S_{|K|}.$
\State Let $f: E(H) \to E(G)$ be the corresponding edge injection (see Subsection \ref{sec:contraction}).
\State Initialize $C'$ as empty (set of vertex disjoint cycles on vertices of $H$)
\For{$1 \le i \le |K|$}
	\For{$i < j \le |K|$}
		\If{$H$ has edge $ij$ at least twice and none of $i, j$ used in $C'$ yet} add cycle $ij$ of length $2$ to $C'$.
		\EndIf
	\EndFor
\EndFor
\For{$1 \le i \le |K|$}
	\If{$i$ not used in $C'$ and $i$ has a self-loop in $H$} add the self-loop $i$ to $C'$.
	\EndIf
\EndFor
\State \Return \PullUp$(H, C', K, \{S_i\}_{i=1}^{|K|}, f)$.
\end{algorithmic}
\label{algo:OneRoundShortCycle}
\end{algorithm}

\begin{algorithm}[h]
  \caption{\ImprovedShortCycle, takes a graph $G$ with $n$ vertices,
    $m = 10n$ vertices, and maximum degree $\Delta$. Returns a
    vertex disjoint cycles containing at least $\frac{m}{10\Delta}$ vertices.}
  \begin{algorithmic}[1]
\State Initialize $C$ as empty ($C$ is the set of cycles we've found).
\If{$n \le 100$} \Return \NaiveShortCycle($G$).
\EndIf
\If{$C$ contains at least $\frac{m}{10\Delta}$ vertices} \Return $(G, C).$ \label{line:checkc}
\EndIf
\State Define $R \assign $ \LowDiamDecomp($G, \frac{1}{12}$).
\State Let $H_1, H_2, \dots, H_k$ be the connected components of $E(G) \backslash R.$
\For{$1 \le i \le k$} let $C_i \assign$ \OneRoundShortCycle($H_i$), where $C_i$ is a set of cycles.
\EndFor
\For{$1 \le i \le k$} $C \assign C \cup C_i$, delete vertices from $C_i$ from $G$.
\EndFor
\State Return to line \ref{line:checkc}.
\end{algorithmic}
\label{algo:ImprovedShortCycle}
\end{algorithm}

We start by analyzing the runtime and guarantees of Algorithm \ref{algo:OneRoundShortCycle}.
\begin{lemma}
\label{lemma:oneround}
 When given a graph $G$ with $n$ vertices, $m$ edges, and maximum degree $\Delta$, Algorithm \ref{algo:OneRoundShortCycle} returns vertex disjoint cycles of length $O(\log n)$ containing at least $\frac{\max(0, m-5n)}{10\Delta\sqrt{m}}$ vertices. It runs in time $O(m).$
\end{lemma}

\begin{proof}
  First, we claim that $|K| \le \frac{1}{2} \sqrt{m}.$ This follows
  from the guarantees of {\TreeSplit} (Lemma \ref{treesplit}) used on
  line \ref{line:treesplitnrtn}. In order to apply Lemma
  \ref{treesplit} we first must check that $\sum_v c_v \ge 4\sqrt{m}.$
  This is clear though, as $\sum_v c_v = 2m \ge 4\sqrt{m}.$ By Lemma
  \ref{treesplit}, we know that the sum of labels in each $K_i$ is at
  least $4\sqrt{m}$, while the sum of labels over all $K_i$ is at most
  $\sum_v c_v = 2m$. Therefore, $|K| \cdot 4\sqrt{m} \le 2m$, so
  $|K| \le \frac{1}{2} \sqrt{m}.$

Also, by the guarantees of {\TreeSplit} (Lemma \ref{treesplit}) and the construction of graph $H$, we know that every vertex of $H$ has degree at most $\Delta \cdot 4\sqrt{m} + \Delta \le 5\Delta \sqrt{m}$. So $\Delta(H) \le 5\Delta \sqrt{m}.$

Next, we show that {\OneRoundShortCycle} indeed satisfies its
guarantee of removing cycles of length $O(\log n)$ containing at least
$\frac{\max(0, m-5n)}{10\Delta \sqrt{m}}$ vertices. Throughout we
assume that $m \ge 5n$, or else the claim is obvious.  It is clear
that $C'$ (as defined in \OneRoundShortCycle) must be a maximal
collection of vertex disjoint cycles of length $1$ and $2$. In other
words, it cannot be enlarged only by adding in new $1$ and
$2$-cycles. Then, we compute the number of edges touching at least
some vertex of $C'.$ If $C'$ involves $t$ vertices, at most
$\Delta(H) t \le 5\Delta \sqrt{m} t$ edges touch some vertex involved
in $C'$. By the pigeonhole principle, any graph with at most
$\frac{1}{2} \sqrt{m}$ vertices and at least $\frac{m}{2}$ edges must
have either a $1$ or $2$-cycle. Therefore, by maximality, we have that
\[ 5\Delta \sqrt{m} t \ge \frac{m}{2}-n, \text{ so we have that } t
  \ge \frac{m-2n}{10\Delta \sqrt{m}} \] as desired.  Combining the
previous discussion with the guarantees of {\PullUp} (Lemma
\ref{pullup}) shows that {\OneRoundShortCycle} successfully removes
cycles of length $O(\log n)$ of total length at least
$\frac{m-5n}{10\Delta \sqrt{m}}$. It is clear that
{\OneRoundShortCycle} runs in time $O(m).$
\end{proof}

Now we proceed to analyzing Algorithm \ref{algo:ImprovedShortCycle}.
\begin{lemma}
\label{lemma:runtimeimprove}
When given a graph $G$ with $n$ vertices, $m = 10n$ edges, and maximum degree $\Delta$, Algorithm \ref{algo:ImprovedShortCycle} with high probability returns vertex disjoint cycles of length $O(\log n)$ containing at least $\frac{m}{10\Delta}$ vertices. It runs in $O(m\sqrt{n})$ time.
\end{lemma}

\begin{proof}
We show that each iteration of Algorithm
\ref{algo:ImprovedShortCycle} also removes cycles of total length at
least $\Omega(\frac{m}{\Delta \sqrt{n}}).$ Indeed, if the algorithm
has not terminated yet (see line \ref{line:checkc}), then we have removed at most $\frac{m}{10\Delta}$
total vertices. Therefore, the graph $G$ will still have at least
$m - \frac{m}{10\Delta} \cdot \Delta = \frac{9}{10}m$ edges remaining.
Additionally, after taking into account the $\frac{m}{12}$ edges from
using \LowDiamDecomp, we see that \[ \sum_{i=1}^k |E(H_i)| \ge \frac{9}{10}m - \frac{1}{12}m \ge \frac{4}{5}m.\]
By Lemma \ref{lemma:oneround}, we get cycles of
length $O(\log n)$ covering at least
\begin{align*}
  &\sum_{i=1}^{|K|} \frac{\max(0, |E(H_i)| - 5|V(H_i)|)}{10\Delta \sqrt{|E(H_i)|}}\\ \ge &\sum_{i=1}^{|K|} \frac{|E(H_i)| - 5|V(H_i)|}{100\Delta\sqrt{n}}\\ \ge &\frac{\frac{4}{5}m - 5n}{100\Delta \sqrt{n}}\\ \ge &\frac{m}{500\Delta \sqrt{n}}
\end{align*}
  total vertices. Here we have used that $m = 10n$.

Therefore, we return to line \ref{line:checkc} of Algorithm \ref{algo:ImprovedShortCycle} at most $O(\sqrt{n})$ times. Each iteration takes $O(m)$ time, for a total runtime of $O(m \sqrt{n})$ as desired.
\end{proof}

%% file: proofs.tex
\section{Main Algorithm and Analysis}
\label{sec:lemmas}

Before continuing, we will state Lemma \ref{sparsify}, whose proof we also defer to the appendix.
We need this to ensure that the graphs we pass to lower levels of the recursion will be sparse
and have bounded maximum degree.
\begin{restatable}{lemma}{restatesparsify}
\label{sparsify}
Let $G$ be a graph with $n$ vertices and $m$ edges. Let $k$ be an integer such that $m \ge k.$
Then algorithm {\Sparsify}$(G, k)$ (Algorithm \ref{algo:sparsify}) returns a subgraph $G' \subseteq G$ with
$n$ vertices, $k$ edges, and $\Delta(G') \le \frac{(2k+4n)\Delta(G)}{m}.$
\end{restatable}

Now, we proceed to our algorithm and analysis. In Algorithm \ref{algo:ShortCycle}, let $\hat{n}, \hat{m}$ denote the number of vertices and edges of the graph at the top level of the recursion.

\begin{algorithm}[p]
\caption{\ShortCycleDecomp, takes a graph $G$ with $n$ vertices, $m = 10n$ edges, max degree $\Delta$, depth $d$ of the recursion (starts at $0$), constant $k = \hat{n}^\frac{1}{c+1}$ \\
Returns vertex disjoint cycles containing at least $\frac{m}{10\Delta}$ vertices}
\begin{algorithmic}[1]
\If{$d = c-1$} \Return \ImprovedShortCycle($G$)
\EndIf
\State Initialize $C \assign \emptyset$ (our set of cycles found so far).
\While{$C$ contains less than $\frac{m}{10\Delta}$ total vertices} \label{line:start}
  \State Let $R \assign$ \LowDiamDecomp($G$, $\frac{1}{12}$).
  \State Let $(A_1, \dots, A_{\ell_1}, B_1, \dots, B_{\ell_2})$ be the connected components of $G\backslash R$, where $|V(A_i)| \le k$ for $1 \le i \le \ell_1$ and $|V(B_i)| > k$ for $1 \le i \le \ell_2.$ \label{line:comp}
  \If{$\sum_{i=1}^{\ell_1} |E(A_i)| \ge \frac{m}{4}$} $C \assign C
  \bigcup_i$ \NaiveShortCycle($A_i$). Go to line \ref{line:start}. \label{line:smallcase}
  \EndIf
  \State For $1 \le i \le \ell_2$, let $T_i$ be a spanning tree of $B_i$ of diameter $O(\log n).$
  \State Initialize $K \assign \emptyset$ ($K$ is a set of subsets of $V(G)$).
  \For{$1 \le i \le \ell_2$}
    \For{$v \in V(B_i)$} set $c_v$ to be the degree of $v$ in $B_i$.
    \EndFor
    \State $K \assign K \cup$ \TreeSplit$(|V(B_i)|, k, T_i, \Delta, \{c_v\}_{v\in V(B_i)})$. \label{line:treesplit}
  \EndFor
  \State Let $K = \{K_1, K_2, \dots, K_{|K|} \}$ (where $K_i \subseteq V(G)$).
  \For{$1 \le i \le |K|$}
    \State Let $S_i$ be a spanning tree of $K_i$ of diameter $O(\log n).$
  \EndFor
  \State Let $H$ be the graph obtained by contracting the components
  $K_1, K_2, \dots, K_{|K|}$
  \State Remove the edges in $H$ corresponding to the edges in the trees
  $S_1, \dots, S_{|K|}.$
  \State Let $f: E(H) \to E(G)$ be the corresponding edge injection (see Subsection \ref{sec:contraction}).
  \State Add isolated vertices to $H$ until $H$ has $\frac{20n}{k}$ vertices. \label{line:addvert}
  \State Let $H' \assign$ \Sparsify($H, \frac{20m}{k}$). \label{line:hprime}
  \State Let $f':E(H') \to E(G)$ be the restriction of $f$ from $E(H)$ to $E(H').$
  \State Let $C' \assign$ \ShortCycleDecomp($H', d+1, k$).
  \State Let $C \assign C \cup$ \PullUp$(H', C', K, \{S_i\}_{i=1}^{|K|}, f')$.
  \For{vertices $v$ part of a cycle in $C$} delete $v$ from $G$.
  \EndFor
\EndWhile
\State \Return $C$.
\end{algorithmic}
\label{algo:ShortCycle}
\end{algorithm}

We now analyze Algorithm {\ShortCycleDecomp}.

We first analyze the case in line \ref{line:smallcase}, where many
edges are within the components $A_i$ (where $|V(A_i)| \le k$).
\begin{lemma}
  \label{smallcase}
  Consider running Algorithm {\ShortCycleDecomp} on a graph $G$ with
  $n$ vertices, $m = 10n$ edges and maximum degree $\Delta$.  In line
  \ref{line:smallcase}, in the case that
  $\sum_{i=1}^{\ell_1} |E(A_i)| \ge \frac{m}{4}$, we can extract
  vertex disjoint cycles of length $O(\log n)$ containing at least
  $\frac{m}{20\Delta}$ vertices in $O(mk)$ time.
\end{lemma}

\begin{proof}
By Lemma \ref{naivecycle}, we know that in a component $C$, we can find vertex disjoint cycles of length $O(\log n)$ containing at least $\frac{|E(C)| - 2|V(C)|}{\Delta}$ vertices in time $O\left(|E(C)| \cdot |V(C)|\right)$. Recall that the components $A_1, A_2, \dots, A_{\ell_1}$ in line \ref{line:comp} of Algorithm \ref{algo:ShortCycle} satisfy $|V(A_i)| \le k.$ Then in total we can find cycles of length $O(\log n)$ containing at least the following number of vertices:
\[ \sum_{i = 1}^{\ell_1} \frac{|E(A_i)| - 2|V(A_i)|}{\Delta} \ge \frac{m/4 - 2n}{\Delta} \ge \frac{m}{20\Delta} \] after using that $n = \frac{m}{10}.$

The total runtime is \[ \sum_{i = 1}^{\ell_1} O\left(|E(A_i)| \cdot k \right) = O(mk) \] as desired.
\end{proof}

Therefore, {\ShortCycleDecomp} will process line \ref{line:smallcase}
at most two times, because after that we would certainly have
constructed cycles containing at least
$\frac{m}{20\Delta} \cdot 2 = \frac{m}{10\Delta}$ vertices by Lemma
\ref{smallcase}. From now on, we assume that the condition of line
\ref{line:smallcase} is false, so
$\sum_{i=1}^{\ell_1} |E(A_i)| < \frac{m}{4}.$

We now bound the number of vertices in our contracted graph $H$ to show that the size of the graph passed down to lower levels of the recursion indeed decreases significantly.
\begin{lemma}
\label{kbound}
  Consider running Algorithm \ref{algo:ShortCycle} on a graph
  $G$ with $n$ vertices, $m = 10n$ edges and maximum degree $\Delta$.
  As defined in Algorithm \ref{algo:ShortCycle}, we have that with high probability $|K| \le \frac{20n}{k}.$ Therefore, after line \ref{line:addvert}, graph $H$ will have exactly $\frac{20n}{k}$ vertices.
\end{lemma}

\begin{proof}
This essentially follows from the guarantees of {\TreeSplit} (Lemma
\ref{treesplit}) as used in line \ref{line:treesplit}. To apply Lemma
\ref{treesplit}, we must show that in each component $B_i$,
we have $\sum_{v \in V(B_i)} c_v \ge k$. Because $B_i$ is connected, we know that
\[ \sum_{v \in V(B_i)} c_v \ge 2(|V(B_i)|-1) \ge k \] as
$|V(B_i)| > k$ by assumption.  Thus, by the guarantees of {\TreeSplit}
(Lemma \ref{treesplit}) used on line \ref{line:treesplit} of Algorithm
\ref{algo:ShortCycle}, we know that each $K_i$ will have total sum of
labels at least $k$. The sum of labels over all $K_i$ is at most the
sum of degrees of $G$, which equals $2m = 20n.$ Therefore,
$k \cdot |K| \le 20n$, so $|K| \le \frac{20n}{k}.$
\end{proof}

Additionally, we need to show that the maximum degree of our graphs doesn't increase significantly from one recursion level to the next.
\begin{lemma}
\label{maxdegree}
Consider running Algorithm \ref{algo:ShortCycle} on a graph
  $G$ with $n$ vertices, $m = 10n$ edges and maximum degree $\Delta$.
  Consider graph $H'$ as defined in line \ref{line:hprime}.
  Then we have with high probability that $\Delta(H') \le 110\Delta.$
\end{lemma}

\begin{proof}
First, we claim that each time we return to line \ref{line:start} of Algorithm \ref{algo:ShortCycle},
the graph $G$ still contains at least $\frac{9}{10}m$ edges. Indeed, we return to line \ref{line:start}
only if we have only created cycles containing at most $\frac{m}{10\Delta}$ vertices.
Therefore, the number of edges $G$ still has among the remaining vertices is at least
$m - \frac{m}{10\Delta} \cdot \Delta = \frac{9}{10}m.$

Next, we show that $\Delta(H) \le (k+1) \Delta(G) = (k+1)\Delta.$ This
essentially follows from the guarantees of {\TreeSplit} (Lemma
\ref{treesplit}). To apply Lemma \ref{treesplit}, must show that in
each component $B_i$, we have $\sum_{v \in V(B_i)} c_v \ge k$. Because
$B_i$ is connected, we know that
\[ \sum_{v \in V(B_i)} c_v \ge 2(|V(B_i)|-1) \ge k \] as
$|V(B_i)| > k$ by assumption.  Thus, by the construction of $H$ in
Algorithm \ref{algo:ShortCycle}, we can see that the degree of vertex
$h \in H$ equals the sum of the labels of the vertices in $K_h.$ By
the guarantees of {\TreeSplit} (Lemma \ref{treesplit}) used on line
\ref{line:treesplit} of Algorithm \ref{algo:ShortCycle}, we can see
that the sum of labels of $K_h$ is at most
$\Delta k + \Delta = \Delta(k+1).$

If we are to get to line \ref{line:hprime}, then we must have skipped
over line \ref{line:smallcase}.  Therefore, we know that
$\sum_{i = 1}^{\ell_2} |E(B_i)| \ge \frac{9}{10}m - \frac{1}{12}m -
\frac{1}{4}m \ge \frac{5}{9}m.$
Because we are removing the edges in $S_1, S_2, \dots, S_{|K|}$ before contracting to get $H$, we have
that
\[ |E(H)| \ge \sum_{i = 1}^{\ell_2} |E(B_i)| - n \ge \frac{4}{9}m. \]
By Lemma \ref{sparsify}, we know that
\begin{align*}
    \Delta(H') &\le \frac{2 \cdot \frac{20m}{k} + 4 \cdot
    \frac{20n}{k}}{4m/9} \cdot \Delta(H)\\ &\le \frac{2 \cdot
    \frac{20m}{k} + 4 \cdot
    \frac{20n}{k}}{4m/9} \cdot \Delta(k+1)\\ &\le 110 \Delta
\end{align*}
\end{proof}

Now, Lemma \ref{maxdegree} allows us to bound the total number of edges processed per level in the recursion.
\begin{lemma}
\label{nextlevel}
Consider running Algorithm \ref{algo:ShortCycle} on a graph
  $G$ with $n$ vertices, $m = 10n$ edges and maximum degree $\Delta$.
  The number of edges processed in level $\ell$ of the recursion with high probability
  is $O(110^{\ell} \hat{m}).$
\end{lemma}

\begin{proof}
We go by induction and show that the number of edges processed
in the next level is at most $110$ times the number of edges in
the previous level. By the guarantees of algorithm {\PullUp} (Lemma \ref{pullup}) and the
recursive guarantees of Algorithm \ref{algo:ShortCycle},
we can see that during each iteration of the algorithm the total length of cycles in $C$ (our set of cycles)
will increase by at least $\frac{|E(H')|}{10 \Delta(H')} \ge \frac{20m/k}{10 \cdot 110 \Delta} \ge \frac{m}{55 \Delta k}.$ Here we used Lemma \ref{maxdegree} to show that $\Delta(H') \le 110 \Delta.$
Therefore, the algorithm will return to line \ref{line:start} at most $\frac{11}{2} k$ times on the same node in the recursion tree, because after that we will have removed at least $\frac{m}{55 \Delta k} \cdot \frac{11}{2} k = \frac{m}{10\Delta}$ vertices. As each $H'$ has $\frac{20m}{k}$ edges,
the total number of edges passed to the next level is at most $\frac{20m}{k} \cdot \frac{11}{2} k = 110m$ as desired.
\end{proof}

We now are ready to state a bound on the runtime of Algorithm \ref{algo:ShortCycle}.
\begin{theorem}
  \label{thm:runtime}
  For all integers $c$, Algorithm \ref{algo:ShortCycle}
  takes as input a graph $G$ with $\hat{n}$ vertices, $\hat{m} = 10\hat{n}$ edges, and maximum degree $\Delta$,
  and with high probability returns vertex disjoint cycles of length $O(\log \hat{n})^c$ containing at least
  $\frac{\hat{m}}{10\Delta}$ vertices. The runtime is $O\left(\hat{m} \hat{n}^\frac{1}{c+1} \cdot 500^c \right)$.
\end{theorem}

\begin{proof}
First, note that per recursion level, by Lemma \ref{smallcase}, we will only perform the computation
listed on line \ref{line:smallcase} at most twice. The computation up to this point takes time $\tilde{O}(m+n)$
by Lemma \ref{lowdiamdecomp}.

Hence, by Lemma \ref{nextlevel}, we know that total number of edges
processed on the $\ell$-th level of recursion is
$O(110^\ell \hat{m}).$ Therefore, the total runtime contribution from
running {\NaiveShortCycle} on small components (see Lemma
\ref{smallcase}) is $O(110^{c-1} \hat{m}k)$. By Lemma
\ref{lemma:runtimeimprove}, the cost of running {\ImprovedShortCycle}
on the bottom level (which is level $c-1$) is at most
\begin{align*}
  &O\left(110^{c-1} \hat{m} \cdot \sqrt{\frac{20^{c-1}
  \hat{n}}{k^{c-1}}} \right)\\ = &O\left(\frac{500^c \hat{m}
  \sqrt{\hat{n}}}{k^\frac{c-1}{2}} \right)\\ = &O(500^c \cdot \hat{m}
  k)
\end{align*}
where we used that $k = \hat{n}^\frac{1}{c+1}.$
The cost of processing the graphs (running the non-recursive
steps of the algorithm) on the $i$-th level is $O(110^i \hat{m}k)$,
which sums to $O(110^{c-1} \hat{m}k)$ in total over all levels.

Therefore, the total runtime is then \[ O\left(500^c \cdot \hat{m} k + 110^c \hat{m} k\right) = O\left(\hat{m}\hat{n}^\frac{1}{c+1} \cdot 500^c \right) \] as desired.

By Lemma \ref{pullup}, at each level, the cycle lengths grow by a factor of $O(\log n).$
Therefore, the total length at the end will be $\paren{O(\log \hat{n})}^{c}$ as desired.
\end{proof}

We can now complete the proof of Theorem
\ref{thm:runtimeandreduction}.
\begin{proof}[Proof of Theorem \ref{thm:runtimeandreduction}]
Note the Theorem \ref{thm:runtime} implies that Algorithm \ref{algo:ShortCycle} satisfies the constraints of Lemma \ref{lemma:reduction}.
Also, Theorem \ref{thm:runtime} tells us that Algorithm \ref{algo:ShortCycle} runs time $O(500^c \cdot \hat{m}\hat{n}^\frac{1}{c+1}) = O(500^c \cdot \hat{n}^\frac{c+2}{c+1})$ because $\hat{m} = 10 \hat{n}.$

Therefore, combining Theorem \ref{thm:runtime} and Lemma \ref{lemma:reduction} tells us that we have an algorithm
that returns a $(20\hat{n}, \paren{O(\log \hat{n})}^c)$-short cycle decomposition which runs in time $O\left(\frac{\hat{m} \cdot 500^c \cdot \hat{n}^\frac{c+2}{c+1}}{\hat{n}} \right) = O(500^c \cdot \hat{m} \hat{n}^\frac{1}{c+1})$ as desired.
\end{proof}

%% file: appendix.tex
\section{Omitted Proofs}

In this section we give proofs for various lemmas which we omitted.
\restatenaivecycle*

\begin{proof}
  Consider the algorithm {\NaiveShortCycle} described as Algorithm~\ref{algo:NaiveShortCycle}

\begin{algorithm}[H]
\caption{\NaiveShortCycle, takes a graph $G$ with $n$ vertices, $m$ edges, and maximum degree $\Delta$ and returns vertex disjoint cycles of length $O(\log n)$ containing at least $\frac{m-2n}{\Delta}$ total vertices}
\begin{algorithmic}[1]
\State Initialize $C$ to be empty
\Repeat
\State While G has a vertex $u$ of degree $\le 2$, remove $u$ and edges incident to $u$ from $G$ \label{line:removeedgeNaive}
\State Run BFS from an arbitrary vertex $r$ until first non-tree edge $e$ found
\State Add the cycle formed by $e$ and tree edges in $C$
\State Remove the vertices in the cycle
\State Remove all corresponding edges
\Until{$G$ is empty}
\State
\Return $C$
\end{algorithmic}
\label{algo:NaiveShortCycle}
\end{algorithm}

After line~\ref{line:removeedgeNaive}, we
will get a graph with minimum degree $3$. Thus, after the BFS, we are
guaranteed to have a tree such that no non-leaf vertex of it has less
than 2 children. Thus, when we find a cycle from the non-tree edge
$e$, it is guaranteed to have length at most $2 \log n$.

Note that when the algorithm ends, we have removed all the
vertices. The only way we remove a vertex not in a cycle is by line 2,
where that vertex has at most 2 edges incident to it when we remove
it. Thus, line 2 can remove at most $2n$ edges incident in
total. Thus, the total number of edges removed by removing vertices in
cycles are at least $m - 2n$. Since the maximum degree of a vertex is
at most $\Delta,$ the number of vertices contained in the cycles must
be at least $\frac{m - 2n}{\Delta}$.  The BFS run in each iteration of
the loop runs in $O(n)$ time since we stop the BFS when the first
non-tree edges is found. Since the loop can run at most $O(n)$ times,
the time taken by the BFS over all iterations is $O(n^{2}).$ Removing
the edges incident to the cycle vertices requires a total time
$O(m+n)$ over all iterations, giving a total running time of
$O(m+n^{2}).$
\end{proof}

\restatepullup*
\begin{proof}
\begin{algorithm}[H]
Throughout, we used indices $\pmod{k}$ where it is obvious.

\caption{\PullUp, takes as inputs a graph $H'$, vertex disjoint cycles $C'$ on the vertices on $H'$, a partition $\{K_i\}_{i=1}^{|V(H')|}$ of the vertices of $G$, a vertex disjoint union of spanning trees $\{S_i\}_{i=1}^{|V(H')|}$ on another graph $G$, and a function $f: E(H') \to E(G)$ which satisfies $f(uv) \in E(K_u, K_v).$ Returns vertex disjoint cycles on the vertices on $G$.}
\begin{algorithmic}[1]
\State Initialize $C$ as empty (ending set of cycles on $G$).
\For{cycle $v_1 v_2 \dots v_k \in C'$}
  \For{$1 \le i \le k$} let $a_i b_{i+1} \assign f(v_i v_{i+1})$.
  \EndFor
  \For{$1 \le i \le k$} let $p_i$ be the path from $b_i \to a_i$ in tree $S_i.$
  \EndFor
  \State $C \assign C \cup b_1 p_1 a_1 b_2 p_2 a_2 \dots b_k p_k a_k$ (concatenation of paths).
\EndFor
\Return $C$
\end{algorithmic}
\label{algo:PullUp}
\end{algorithm}

All guarantees follow very easily from the description of Algorithm \ref{algo:PullUp}.
At a high level, note that by the definition of $H'$ and $K_i$, we know that cycles on $H'$ corresponds to ``cycles" on the components $K_i$
in the graph $G$. Now, simply use the edges of the spanning trees $S_i$ to recover a cycle on $G$.

It is obvious that we cover at least as many vertices among $C$ as in $C'$. Additionally, the lengths will increase by at most a factor of the diameter of some spanning tree $S_i$ by the construction. Vertex disjointness follows trivially. The runtime follows because the only operation we need to do is find paths between vertices in a spanning tree, which is time $O(n).$
\end{proof}

\restatetreesplit*

\begin{proof}
Let the vertices be numbered $1, 2, \dots, n$, and let the corresponding labels be $c_1, c_2, \dots, c_n.$
We outline the algorithm \TreeSplit, which takes $n, t, T, X$, and $\{c_i\}_{i = 1}^n$ as inputs.
\begin{algorithm}[H]
\caption{\TreeSplit, takes inputs $n, t, T, X$, and $\{c_v\}_{v \in V(T)}$ such that $\sum_i c_i \ge t$. Runs in time $O(n)$.
Partitions $T$ into connected subtrees so that each subtree has sum of labels between $t$ and $Dt+X$.}
\begin{algorithmic}[1]
\State Root $T$ at a leaf $\ell$.
\State Initialize an array $\extra[1 \dots n]$
\State Set $\extra[v] \assign c_v$ for all vertices $v$.
\State Run a depth first search through $T$, starting at $\ell$.
\State Let $v$ denote the vertex that is being currently processed.
\For{children $u$ of $v$}
	\State Recursively visit $u$.
	\State $\extra[v] \assign \extra[v] + \extra[u]$.
\EndFor
\If{$\extra[v] \ge t$}
	\State Remove the edge between $v$ and its parent.
	\State Set $\extra[v] \assign 0.$
\EndIf
\State End the depth first search.
\If{$\extra[\ell] < t$}
	\State Reconnect the component with $\ell$ to another component.
\EndIf
\State Let $C_1, C_2, \dots, C_k$ be the connected components of the resulting forest.
\State \Return $(C_1, C_2, \dots, C_k).$
\end{algorithmic}
\label{algo:treesplit}
\end{algorithm}

In the algorithm, $\extra[v]$ denotes the total sum of labels still attached to the rest of the tree $v$ after processing it.

We can show the correctness of the algorithm by induction, specifically that after processing vertex $v$, that $\extra[v] < t$ always.
Consider a vertex $v$. Note that it has at most $D-1$ children $u_1, \dots, u_k$ (as we rooted the tree at a leaf $\ell$). By induction it is clear that after visiting all children of $v$, it will be true that
\[ \extra[v] \le c_v + \sum_i \extra[u_i] \le (D-1)t + X. \]
If $t \le \extra[v] \le (D-1)t + X$, then we split off $v$ from its parent (we made a new component). Otherwise, $\extra[v] < t$, verifying the induction.
Finally, if the root $\ell$ satisfies $\extra[\ell] < t$ but $\extra[\ell] \neq 0$, then connect the component containing $\ell$ to another one (which we know has label sum at most $(D-1)t+X$). Therefore, it still holds that all components have sum of labels at most $t + (D-1)t + X = Dt + X.$

This shows the correctness of Algorithm \ref{algo:treesplit}.
\end{proof}

\restatesparsify*
\begin{proof}
Consider the following algorithms:
\begin{algorithm}[H]
\caption{\SparsifyHelper, takes the graph $G$ and a spanning tree root $r$, recursively remove the tree edge from leaves to the root with odd degree vertices.}
\begin{algorithmic}[1]
\For{$v$ be child of $r$}
\State \textsc{SparsifyHelper}$(G, v)$
\EndFor
\If{$r$ has odd degree}
\State remove the edge formed by $r$ and parent of $r$
\EndIf
\end{algorithmic}
\end{algorithm}
\begin{algorithm}[H]
\caption{\Sparsify, takes a graph $G$ with $m$ vertices, $n$ edges, and max degree $\Delta$, and a target edge count $k$.
Returns a subgraph with exactly $k$ edges and max degree at most $\frac{(2k+4n)\Delta}{m}$}
\begin{algorithmic}[1]
\While{$|E(G)| \ge 2k+2n$}
	\State For each connected component of $G$, construct a spanning tree rooted at an arbitrary vertex $r_i$ of that component.
	\For{Each root $r_i$}
	\State \textsc{SparsifyHelper}$(G, r_i)$
	\EndFor
	\State Perform a Eulerian Tour on each connected component, and remove every other edge starting from the first edge. Specifically, if the Eulerian tour has edges $e_1, e_2, \dots, e_\ell$ in that order, remove edges $e_1, e_3, e_5, \dots$
\EndWhile
\State Remove edges until the resulting graph has exactly $k$ edges.
\end{algorithmic}
\label{algo:sparsify}
\end{algorithm}

First we claim that the remaining graph after \textsc{SparsifyHelper}
will have even degree for all vertices. Since we are removing the tree
edge when the degree of the vertex is odd from the bottom to the root,
we are guaranteed that all vertices except the root will have an even
degree. But the sum of degrees of all vertices is even, so the root
will have an even degree as well. Thus, we can perform the Eulerian
Tour on each connected component . After removing every other edge,
every degree will be reduced by half. Note that in components with an
odd number of edges, we must remove both the first and last edges in
the Eulerian tour. In total, if we remove $t$ edges from the call to
{\SparsifyHelper} and that after doing this, we have $o$ components of
odd number of edges remaining. It is easy to see that our resulting graph after
deleting every other edge from the Eulerian tours on connected
components has $\frac{m-t-o}{2} \ge \frac{m}{2}-n$ edges.

We claim that after $i$th round,
$\frac{m}{2^i} - 2n \le |E_i| \le \frac{m}{2^i}$. This can be easily
shown by induction. After the $(i+1)$th round,
\begin{align*}
|E_{i+1}| &\ge \frac{|E_i|}{2} - n \ge \frac{m/2^i - 2n}{2} - n = \frac{m}{2^{i+1}}-2n \\
|E_{i+1}| & \le \frac{|E_i|}{2} \le \frac{m/2^i}{2} = \frac{m}{2^{i+1}}
\end{align*}
Also, when the algorithm ends, we have $k \le |E| < 2k+2n$. Assume
that the
algorithm needs $r$ rounds. Then, we have $k \le \frac{m}{2^r}$ and
$\frac{m}{2^r} - 2n < 2k+2n$. So $\frac{m}{2^r} < 2k + 4n$. Thus, the
number of rounds is in the range of
$(\log_2\frac{m}{2k+4n}, \log_2\frac{m}{k}]$. Therefore, our final
graph $G'$ will have
\[ \Delta(G') <
  \frac{\Delta(G)}{2^{\log_2\frac{m}{2k+4n}}}=\frac{(2k+4n)\Delta(G)}{m}. \]
To analyze the running time, note that the $i$-th round takes time
$O(\frac{m}{2^i} + n)$ time. As we are running for $O(\log n)$ rounds,
our total runtime will be
$\sum_{i = 0}^{O(\log n)} O(\frac{m}{2^i} + n) = O(m + n \log n)$ as
desired.
\end{proof}